\documentclass[journal]{IEEEtran}
\usepackage{amsfonts}
\usepackage{amsmath}
\usepackage{amssymb}
\usepackage{epsfig}
\usepackage{graphics}
\usepackage{graphicx}
\usepackage{needspace}
\usepackage{color, colortbl}
\definecolor{Gray}{gray}{0.9}
\epsfclipon
\newtheorem{theorem}{Theorem}[section]

\newtheorem{lem}[theorem]{Lemma}

\newtheorem{problem}{Problem}
\newtheorem{assumption}{Assumption}
\newtheorem{proposition}{Proposition}
\newtheorem{remark}{Remark}[section]

\makeatother

\begin{document}

\title{Optimal Vaccine Allocation to Control Epidemic Outbreaks in Arbitrary
Networks}

\author{Victor M. Preciado, Michael Zargham, Chinwendu Enyioha, Ali Jadbabaie, and George Pappas %
\thanks{The authors are with the Department of Electrical and Systems Engineering
at the University of Pennsylvania, Philadelphia PA 19104.  %
}}
\maketitle
\begin{abstract}
We consider the problem of controlling the propagation of an epidemic
outbreak in an arbitrary contact network by distributing vaccination
resources throughout the network. We analyze a networked version of
the Susceptible-Infected-Susceptible (SIS) epidemic model when individuals
in the network present different levels of susceptibility to the epidemic.
In this context, controlling the spread of an epidemic outbreak can
be written as a spectral condition involving the eigenvalues of a
matrix that depends on the network structure and the parameters of
the model. We study the problem of finding the optimal distribution
of vaccines throughout the network to control the spread of an epidemic
outbreak. We propose a convex framework to find cost-optimal distribution
of vaccination resources when different levels of vaccination are
allowed. We also propose a greedy approach with quality guarantees
for the case of all-or-nothing vaccination. We illustrate our approaches
with numerical simulations in a real social network.
\end{abstract}

\section{Introduction}

Motivated by the problem of epidemic spread in human networks, we
analyze the problem of controlling the spread of a disease by distributing
vaccines throughout the individuals in a contact network. The problem
of controlling spreading processes in networks appear in many different
settings, such as epidemiology \cite{Bai75,AM91}, computer viruses
\cite{GGT03}, or viral marketing \cite{LAH07}. The dynamic of the
spread depends on both the structure of the contact network, the epidemic
model and the values of the parameters associated to each individual.
We model the spread using a recently proposed variant of the popular
SIS epidemic model in which the infection rate is allowed to vary
among the set of individuals in the network \cite{MOK09}. In our
setting, we can modify the individual infection rates, within a feasible
range, by injecting different levels of vaccination in each node.
Injecting a particular level of vaccination in a node has also an
associated cost, which can vary from individual to individual. In
this context, we propose efficient convex framework to find the optimal
distribution of vaccination resources throughout the networks.

The dynamic behavior of spreading processes in networks have been
widely studied. In \cite{New02}, Newman studied the epidemic thresholds
on several random graphs models. Pastor-Satorras and Vespignani studied
viral propagation in power-law networks \cite{PV01}. This initial
work was followed by a long list of papers aiming to study the spread
in more realistic network models. Boguna and Pastor-Satorras \cite{BP02}
considered the spread of a virus in correlated networks, where the
connectivity of a node is related to the connectivity of its neighbors.
In \cite{PJ09}, the authors analyze spreading processes in random
geometric networks. The analysis of spreading processes in arbitrary
contact networks was first studied by Wang et al. \cite{WCWF03} for
the case of discrete-time dynamics. In \cite{GMT05}, Ganesh et al.
proposed a continuous-time Markov process to relate the speed of spreading
with the largest eigenvalue of the adjacency matrix of the contact
network. The connection between the speed of spreading and the spectral
radius of the network was also found for a wide range of spreading
models in \cite{CWWLF08}. The relationship between the spectral radius
of a contact network and its local structural properties were explored
in \cite{PJ13,PJD13}.

The development of strategies to control the dynamic of a spread process
is a central problem in public health and network security. In \cite{BCGS10},
Borgs et al. proposed a probabilistic analysis, based on the theory
of contact processes, to characterize the optimal distribution of
a fixed amount of antidote in a given contact network. In \cite{AAITF13},
Aditya et al. proposed several heuristics to immunize individuals
in a network to control virus spreading processes. In the control
systems literature, Wan et al. proposed in \cite{WRS08} a method
to design optimal strategies to control the spread of a virus using
eigenvalue sensitivity analysis ideas together with constrained optimization
methods. Our work is closely related to the work in \cite{GOM11}
and \cite{DC12}, in which a continuous-time time Markov processes,
called the N-intertwined model, is used to analyze and control the
spread of a SIS epidemic model.

In this paper, we propose a convex optimization framework to efficiently
find the cost-optimal distribution of vaccination resources in an
arbitrary contact network. In our work, we use a heterogeneous version
of the N-intertwined SIS model \cite{MOK09} to model a spread process
in a network of individuals with different rate of being infected
and recovered. We assume that we can modify the rates of infection
of individuals, within a feasible range, by distributing vaccines
to the individuals in the network. We assume that there is a cost
associated to injecting a particular amount of vaccination resources
to a each individual, where the cost function can vary from individual
to individual. Our aim is to find the optimal distribution of vaccination
resources throughout the network in order to control the spread of
an initial infection at a minimal cost. We consider two version of
this problem: (\emph{i}) The \emph{fractional case,} in which we are
allowed to inject a fractional amount of vaccination resources in
each node of the network, and (\emph{ii}) the \emph{combinatorial
case,} in which we fully vaccinate a selection of individuals in the
network, leaving the rest of nodes unvaccinated.

The paper is organized as follows. In Section II, we introduce our
notation, as well as some background needed in our derivations. In
Section III, we formulate our problem and provide an efficient solution
based on convex optimization. In Section IV, we study a combinatorial version of the problem studied in Section III and provide a greedy heuristic algorithm with a quality guarantee. We include some conclusions in Section V.

\section{\label{Notation}Notation \& Preliminaries}

In this section we introduce some graph-theoretical nomenclature and
the dynamic spreading model under consideration.

\subsection{Graph Theory}

Let $\mathcal{G}=\left(\mathcal{V},\mathcal{E}\right)$ denote an
undirected graph with $n$ nodes, $m$ edges, and no self-loops%
\footnote{An undirected graph with no self-loops is also called a \emph{simple}
graph.%
}. We denote by $\mathcal{V}\left(\mathcal{G}\right)=\left\{ v_{1},\dots,v_{n}\right\} $
the set of nodes and by $\mathcal{E}\left(\mathcal{G}\right)\subseteq\mathcal{V}\left(\mathcal{G}\right)\times\mathcal{V}\left(\mathcal{G}\right)$
the set of undirected edges of $\mathcal{G}$. If $\left\{ i,j\right\} \in\mathcal{E}\left(\mathcal{G}\right)$
we call nodes $i$ and $j$ \emph{adjacent} (or neighbors), which
we denote by $i\sim j$. We define the set of neighbors of a node
$i\in\mathcal{V}$ as $\mathcal{N}_{i}=\{j\in\mathcal{V}\left(\mathcal{G}\right):\left\{ i,j\right\} \in\mathcal{E}\left(\mathcal{G}\right)\}$.
The number of neighbors of $i$ is called the \emph{degree} of node
i, denoted by $d_{i}$. The adjacency matrix of an undirected graph
$\mathcal{G}$, denoted by $A_{\mathcal{G}}=[a_{ij}]$, is an $n\times n$
symmetric matrix defined entry-wise as $a_{ij}=1$ if nodes $i$ and
$j$ are adjacent, and $a_{ij}=0$ otherwise%
\footnote{For simple graphs, $a_{ii}=0$ for all $i$.%
}. Since $A_{\mathcal{G}}$ is symmetric, all its eigenvalues, denoted
by $\lambda_{1}(A_{\mathcal{G}})\geq\lambda_{2}(A_{\mathcal{G}})\geq\ldots\geq\lambda_{n}(A_{\mathcal{G}})$,
are real.

\subsection{N-Intertwined SIS Epidemic Model}

Our modeling approach is based on the N-intertwined SIS model proposed
by Van Mieghem et at. in \cite{MOK09}. In contrast with previously
proposed models, the N-intertwined model is a continuous-time networked
Markov process with $2^{n}$ states able to model the dynamics of
a viral infection in an arbitrary contact network. Using the Kolmogorov
forward equations and a mean-field approach, one can approximate the
dynamics of the viral spread using a system of $n$ ordinary differential
equations, as follows. Consider a network of $n$ individuals described
by the adjacency matrix $A_{\mathcal{G}}=\left[a_{ij}\right]$. The
infection probability of an individual at node $i\in\mathcal{V\left(G\right)}$
at time $t\geq0$ is denoted by $p_{i}(t)$. Let us assume, for now,
that the viral spreading is characterized by two positive parameters--a
constant infection rate $\beta\geq0$ and a curing rate $\delta\geq0$.
Hence, the N-intertwined SIS model in \cite{MOK09} is described by
the following set of $n$ ODE's:
\begin{equation}
\frac{dp_{i}\left(t\right)}{dt}=\left(1-p_{i}\left(t\right)\right)\beta\sum_{j=1}^{n}a_{ij}p_{j}\left(t\right)-\delta p_{i}\left(t\right),\label{eq:Mieghem}
\end{equation}
for $i=1,\ldots,n$.

As proved in \cite{MOK09}, the exact probability of infection is
upper bounded by its approximation $p_{i}\left(t\right)$. A local
stability analysis of the above system of ODE's around the disease-free
equilibrium, $p_{i}=0$ for all $i$, provides the following result
\cite{MOK09}:
\begin{proposition}
\label{prop:HomogeneousSISstability}Consider the N-intertwined SIS
epidemic model in (\ref{eq:Mieghem}). Then, an initial infection
converge to zero exponentially fast if
\[
\lambda_{1}\left(A_{\mathcal{G}}\right)<\frac{\delta}{\beta}.
\]

\end{proposition}
The above provides a simple condition to guarantee a controlled epidemic
dynamics in terms of the largest eigenvalue of the adjacency matrix.
In the following, we derive a similar condition when the infection
parameters vary from individual to individual within the network. 

\subsection{Non-Homogeneous N-Intertwined SIS Epidemic Model}

A direct extension of the N-intertwined model for node-specific infection
and curing rates, $\beta_{i}$ and $\delta_{i}$, is
\[
\frac{dp_{i}\left(t\right)}{dt}=\left(1-p_{i}\left(t\right)\right)\beta_{i}\sum_{j=1}^{n}a_{ij}p_{j}\left(t\right)-\delta_{i}p_{i}\left(t\right).
\]
We can write the above dynamics in matrix form as 
\begin{equation}
\frac{d\boldsymbol{p}\left(t\right)}{dt}=\left(BA_{\mathcal{G}}-D\right)\boldsymbol{p}\left(t\right)-P\left(t\right)BA_{\mathcal{G}}\boldsymbol{p}\left(t\right),\label{eq:heteroSIS}
\end{equation}
where $\boldsymbol{p}\left(t\right)=\left(p_{1}\left(t\right),\ldots,p_{n}\left(t\right)\right)^{T}$,
$B=diag(\beta_{i})$, $D=diag\left(\delta_{i}\right)$, and $P\left(t\right)=diag(p_{i})$.
Concerning the non-homogeneous epidemic model, we have the following
result:
\begin{proposition}
\label{prop:Heterogeneous SIS stability condition}Consider the heterogeneous
N-intertwined SIS epidemic model in (\ref{eq:heteroSIS}). Then, if
\[
\lambda_{1}\left(BA-D\right)\leq-\varepsilon,
\]
an initial infection $\boldsymbol{p}\left(0\right)\in\left[0,1\right]^{n}$
will converge to zero exponentially fast, i.e., there exists an $\alpha>0$
such that \textup{$\left\Vert p_{i}\left(t\right)\right\Vert \leq\alpha\left\Vert p_{i}\left(0\right)\right\Vert e^{-\varepsilon t}$,
for all $t\geq0$.}\end{proposition}
\begin{IEEEproof}
First, we have 
\begin{align*}
\frac{dp_{i}\left(t\right)}{dt} & =\beta_{i}\sum_{j=1}^{n}a_{ij}p_{j}\left(t\right)-\delta_{i}p_{i}\left(t\right)-\beta_{i}p_{i}\left(t\right)\sum_{j=1}^{n}a_{ij}p_{j}\left(t\right)\\
\leq & \beta_{i}\sum_{j=1}^{n}a_{ij}p_{j}\left(t\right)-\delta_{i}p_{i}\left(t\right),
\end{align*}
since $\beta_{i}$, $\delta_{i}$, $p_{i}\left(t\right)$,$a_{ij}\geq0$.
Therefore, the linear dynamic system
\begin{equation}
\frac{d\hat{p}_{i}\left(t\right)}{dt}=\beta_{i}\sum_{j=1}^{n}a_{ij}\hat{p}_{j}\left(t\right)-\delta_{i}\hat{p}_{i}\left(t\right),\label{eq:SISlinear}
\end{equation}
upper-bounds the nonlinear dynamical system (\ref{eq:heteroSIS})
when they share the same initial conditions, i.e., $\hat{\boldsymbol{p}}\left(t\right)\geq\boldsymbol{p}\left(t\right)$
for $t\geq0$ when $\hat{\boldsymbol{p}}\left(0\right)=\boldsymbol{p}\left(0\right)$.

This linear dynamic system can be written in matrix form as
\[
\frac{d\hat{\boldsymbol{p}}\left(t\right)}{dt}=\left(BA_{\mathcal{G}}-D\right)\hat{\boldsymbol{p}}\left(t\right).
\]
For the above linear system to be stable, we need the eigenvalues
of $BA-D$ to be in the open left half-plane. The state matrix $BA_{\mathcal{G}}-D$
has real eigenvalues, since it can be transform via a similarity transformation
to the symmetric matrix $B^{1/2}A_{\mathcal{G}}B^{1/2}-D$. Hence,
exponential asymptotic stability, with an exponential rate $\varepsilon$,
is equivalent to the largest eigenvalue $\lambda_{1}\left(BA_{\mathcal{G}}-D\right)<-\varepsilon$.
\end{IEEEproof}
In the above analysis, we have shown that the linear dynamics in (\ref{eq:SISlinear})
upper-bounds the mean-field approximation in (\ref{eq:heteroSIS});
thus, the spectral result in Proposition \ref{prop:Heterogeneous SIS stability condition}
is a sufficient condition to control the evolution of an epidemic
outbreak. In the following section, we use this result to characterize
the profiles of infection rates that results in a stable linear dynamics.

\section{A Convex Framework for Optimal Resource Allocation}

Our main aim is to propose an efficient optimization framework to
find the optimal distribution of vaccines to control the spread of
an epidemic outbreak in a given network. In this section, we consider
the fractional vaccination problem. In the fractional case, we assume
that we are able to modify the infections rates $\beta_{i}$ in the
network by distributing vaccination resources throughout the individuals
in the network. We assume that the infection rates of each individual
can be modified within a particular feasible interval, $\underline{\beta}_{i}\leq\beta_{i}\leq\bar{\beta}_{i}$,
where $\bar{\beta}_{i}>0$ is the value of the natural infection rate
for node $i$, which is achieved in the absence of any nodal immunization,
and $\underline{\beta}_{i}>0$ is the minimum possible infection rate
for node $i$, which is achieved when we allocate a large amount of
vaccines at node $i$. In Section \ref{sec:Combinatorial-Resource-Allocatio},
we will consider a combinatorial version of the above fractional strategies.
In the combinatorial case, we will assume that the infection rate
can only take one of two values, $\beta_{i}\in\left\{ \underline{\beta}_{i},\bar{\beta}_{i}\right\} $.
In the fractional case considered in this section, we propose an optimization
framework to find the optimal distribution of resources when there
is a cost function function associated to different values of $\beta_{i}$.

\subsection{Vaccination Cost}

The cost of achieving a particular infection rate for node $i$ is
denoted by $f_{i}\left(\beta_{i}\right)$. This cost function is node-dependent
and presents the following properties:
\begin{enumerate}
\item The cost of achieving the natural infection rate is zero, i.e., $f_{i}\left(\bar{\beta}_{i}\right)=0$.
\item The maximum cost of vaccinating node $i$, denoted by $T_{i}$, is
achieved at the minimum infection rate, i.e., $\max_{\beta_{i}}f_{i}\left(\beta_{i}\right)=f_{i}\left(\underline{\beta}_{i}\right)\triangleq T_{i}$.
\item The vaccination cost function is monotonically decreasing in the interval
$\beta_{i}\in\left[\underline{\beta}_{i},\bar{\beta}_{i}\right]$.
\end{enumerate}
Apart from the above properties, we make the following convexity assumptions
on the cost function $f_{i}$ to obtain a tractable convex framework:

\begin{assumption} The vaccination cost function, $f_{i}\left(\beta_{i}\right)$,
is twice differentiable and satisfies the following constrain:
\begin{equation}
f''_{i}\left(\beta_{i}\right)\geq-\frac{2}{\beta_{i}}f'_{i}\left(\beta_{i}\right),\label{eq:StrongConvexity}
\end{equation}
for $\beta_{i}\in\left[\underline{\beta}_{i},\bar{\beta}_{i}\right]$.
\end{assumption}

Notice that, since $f_{i}$ is monotonically decreasing, we have that
$f'_{i}\left(\beta_{i}\right)<0$; thus, we have that Assumption 1
implies that $f''_{i}\left(\beta_{i}\right)>0$. In other words, Assumption
1 is stronger than convexity. For example, a function that satisfies
Assumption 1 with equality is:
\begin{equation}
f_{i}\left(\beta_{i}\right)=T_{i}\frac{\beta_{i}^{-1}-\bar{\beta}_{i}^{-1}}{\underline{\beta}_{i}^{-1}-\bar{\beta}_{i}^{-1}}.\label{eq:Quasiconvex Limit}
\end{equation}
In practice, for low values of $\underline{\beta}_{i}$ and $\bar{\beta}_{i}$,
this function takes a shape of practical interest. For example. in
Fig. \ref{fig_1} we plot the function in (\ref{eq:Quasiconvex Limit}) for
$\underline{\beta}_{i}=1.75e-3$, $\bar{\beta}_{i}=8.66e-3$, and
$T_{i}=1$. In the abscissa of this plot, we represent the vaccination
cost $f_{i}\left(\beta_{i}\right)$, which is in the range $\left[0,1\right]$.
We observe how the cost function is convex and presents diminishing
returns, since the reduction in the infection rate for a given amount
of investment is greater in the low-cost range than in the high-cost
range.
\begin{figure}[t]
\centering\includegraphics[width=0.45\textwidth]{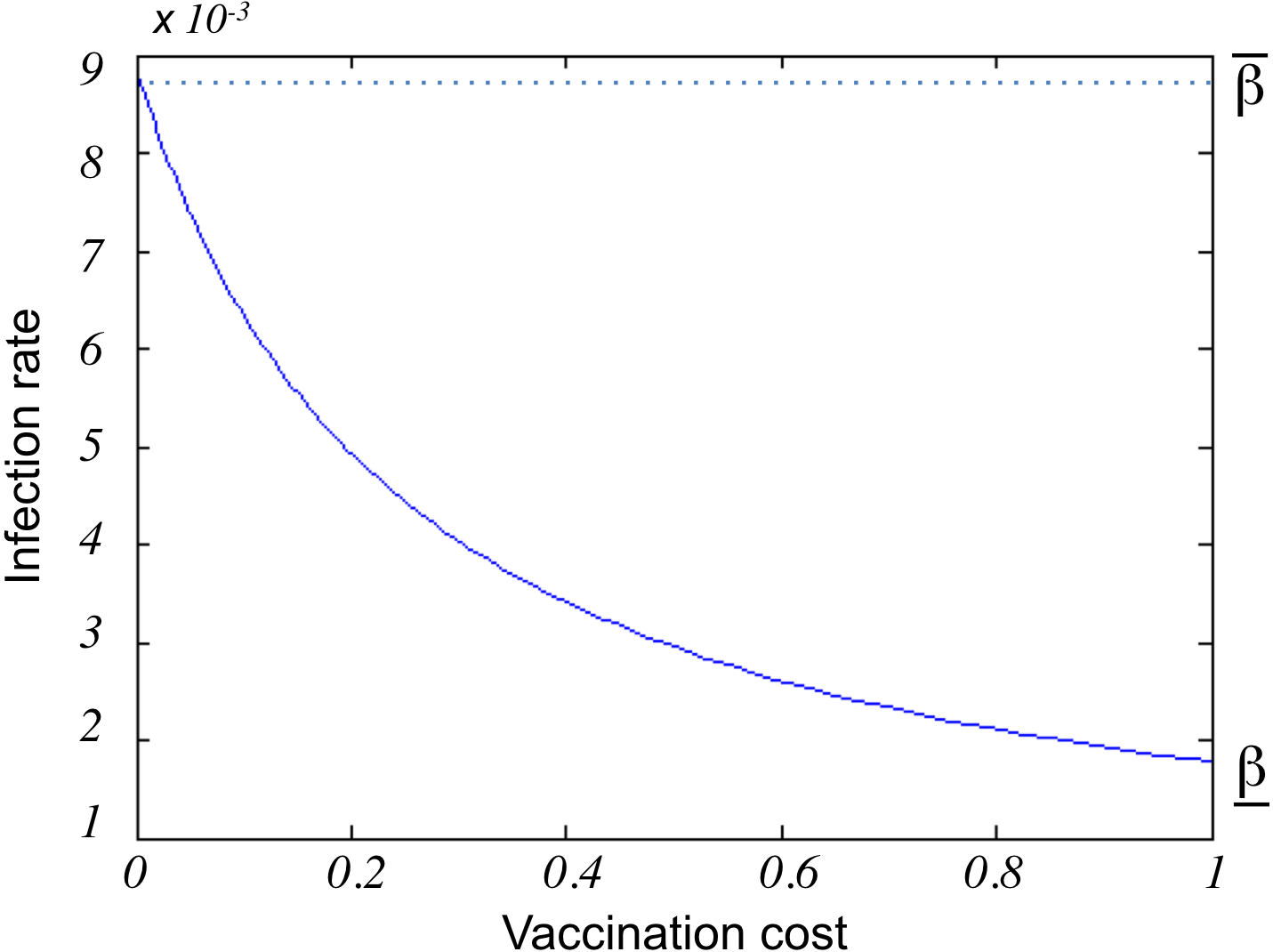}
\caption{Convex cost function in (\ref{eq:Quasiconvex Limit}).}
\label{fig_1}
\end{figure}

\subsection{Problem Statements}

In this subsection we propose an optimization framework to find the
cost-optimal allocation of vaccines in a given contact network $\mathcal{G}$
with adjacency matrix $A_{\mathcal{G}}$. In particular, we consider
the following problem:
\begin{problem}
\emph{\label{Problem:OptimalVaccineDistribution}Given a curing rate
profile, $\left\{ \delta_{i}:i\in\mathcal{V}\left(\mathcal{G}\right)\right\} $,
and a vaccination cost function $f_{i}\left(\beta_{i}\right)$} for
$\beta_{i}\in\left[\underline{\beta}_{i},\bar{\beta}_{i}\right]$,
\emph{find the optimal distribution of vaccines to control the propagation
of an epidemic outbreak with an asymptotic exponential decaying rate
$\varepsilon$ at a total minimum cost.}
\end{problem}
According to Proposition \ref{prop:Heterogeneous SIS stability condition},
this problem can be mathematically stated as the following optimization
problem:
\begin{eqnarray}
T^{*}=min_{\left\{ \beta_{i}\right\} } & \sum_{i=1}^{n}f_{i}\left(\beta_{i}\right)\nonumber \\
s.t. & \lambda_{1}\left(BA_{\mathcal{G}}-D\right)\leq-\varepsilon\label{eq:OptimalVaccine}\\
 & \underline{\beta}_{i}\leq\beta_{i}\leq\bar{\beta}_{i}, & i=1,\ldots,n,\nonumber 
\end{eqnarray}

In the following subsection, we propose a convex formulation to solve
this problem under Assumption 1.

\subsection{Semidefinite Programming (SDP) Approach}

Our formulation is based on writing the spectral stability condition
$\lambda_{1}\left(BA_{\mathcal{G}}-D\right)\leq-\varepsilon$ using
a simple semidefinite constrain. In particular, we have the following
result:
\begin{lem}
For $A_{\mathcal{G}}$ symmetric, $B=diag\left(\beta_{i}\right)$,
and $D=diag\left(\delta_{i}\right)$, we have that \textup{$\lambda_{1}\left(BA_{\mathcal{G}}-D\right)\leq-\varepsilon$
if and only if $\left(D-\varepsilon I\right)B^{-1}-A_{\mathcal{G}}\succeq0$.}\end{lem}
\begin{IEEEproof}
Notice that $BA_{\mathcal{G}}-D$ is a matrix similar to $B^{1/2}A_{\mathcal{G}}B^{1/2}-D$,
since we can pre- and post- multiply the former matrix by $B^{-1/2}$
and $B^{1/2}$, respectively, to obtain the latter. Hence, since $B^{1/2}A_{\mathcal{G}}B^{1/2}-D$
is a symmetric matrix with real eigenvalues, the eigenvalues of $BA_{\mathcal{G}}-D$,
including $\lambda_{1}\left(BA_{\mathcal{G}}\right)$, are all real.
Then, we have that $\lambda_{1}\left(BA_{\mathcal{G}}-D\right)\leq-\varepsilon$
if and only if $\lambda_{i}\left(\left(D-\varepsilon I\right)-BA_{\mathcal{G}}\right)=\lambda_{i}\left(\left(D-\varepsilon I\right)-B^{1/2}A_{\mathcal{G}}B^{1/2}\right)\geq0$,
which is equivalent to $\left(D-\varepsilon I\right)-B^{1/2}A_{\mathcal{G}}B^{1/2}\succeq0$.
Applying a congruence transformation to $\left(D-\varepsilon I\right)-B^{1/2}A_{\mathcal{G}}B^{1/2}$
by pre- and post-multiplying by $B^{-1/2}$, we obtain that $\lambda_{1}\left(BA_{\mathcal{G}}-D\right)\leq-\varepsilon$
if and only if $\left(D-\varepsilon I\right)B^{-1}-A_{\mathcal{G}}\succeq0$.
\end{IEEEproof}
Using the above Lemma, we can rewrite the optimization problem \ref{Problem:OptimalVaccineDistribution}
as a convex optimization program, as follows. First, let us rewrite
(\ref{eq:OptimalVaccine}) using the change of variables $\gamma_{i}\triangleq\beta_{i}^{-1}$as,
\begin{eqnarray}
T^{*}\triangleq\min_{\left\{ \gamma_{i}\right\} } & \sum_{i=1}^{n}f_{i}\left(\gamma_{i}^{-1}\right)\nonumber \\
s.t. & \left(D-\varepsilon I\right)\Gamma-A_{\mathcal{G}}\succeq0\nonumber \\
 & \bar{\beta}_{i}^{-1}\leq\gamma_{i}\leq\underline{\beta}_{i}^{-1}, & i=1,\ldots,n,\label{eq:OptimalConvexVaccination}
\end{eqnarray}
where $\Gamma=diag\left(\gamma_{i}\right)$. Therefore, the feasible
set is convex in the space of variables $\gamma_{i}$, $i=1,\ldots,n$.
Furthermore, we now verify that the cost function $\sum_{i=1}^{n}f_{i}\left(\gamma_{i}^{-1}\right)$
is also convex under Assumption 1 by computing its second derivative,
\[
\frac{d^{2}}{d\gamma_{i}^{2}}\sum_{i}f_{i}\left(\gamma_{i}^{-1}\right)=f''_{i}\left(\gamma_{i}^{-1}\right)\frac{1}{\gamma_{i}^{4}}+2f'_{i}\left(\gamma_{i}^{-1}\right)\frac{1}{\gamma_{i}^{3}}\geq0,
\]
where the last inequality is obtained from Assumption 1, taking
into account that $\gamma_{i}^{-1}=\beta_{i}$.

The convex optimization program in (\ref{eq:OptimalConvexVaccination})
allows us to efficiently find the cost-optimal allocation of vaccines
to control the spread of an epidemic outbreak in a given contact network.
In the following subsection, we illustrate our approach in a real
social network.

\subsection{\label{sub:Numerical-Results}Numerical Results}
\begin{figure*}
\centering\includegraphics[width=1.0\textwidth]{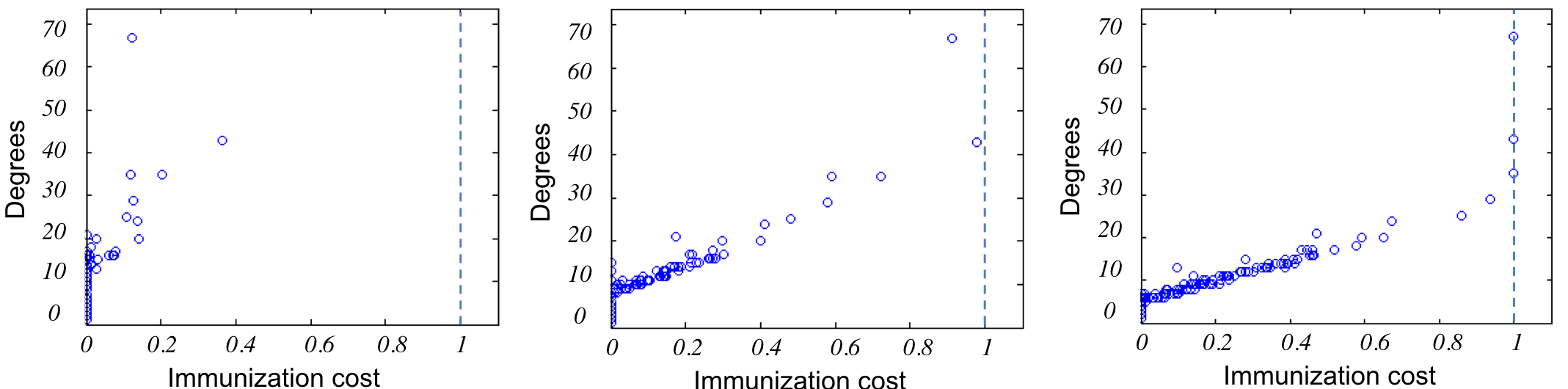}
\caption{Vaccination costs versus degree in a social network with 247 nodes.}
\label{fig_2}
\end{figure*}We illustrate our results by designing the optimal distribution of
vaccines in an online social network when the cost vaccination function
follows (\ref{eq:Quasiconvex Limit}). We consider a social network
with 247 nodes, and assume that the individuals in the network present
the same recovery rate, $\delta_{i}=\delta=0.1$. In this case, we
can rewrite (\ref{eq:OptimalConvexVaccination}) as a convex program
with a convenient structure, as follows. First, defining $a\triangleq\left(\underline{\beta}_{i}^{-1}-\bar{\beta}_{i}^{-1}\right)^{-1}$,
we have that 
\[
\sum_{i}f_{i}\left(\beta_{i}\right)=a\sum_{i}\beta_{i}^{-1}-a\sum_{i}\bar{\beta}_{i}^{-1}=a\text{�}\mbox{Trace}\left(\Gamma\right)-b,
\]
where $b\triangleq\alpha\sum_{i}\bar{\beta}_{i}^{-1}$. Hence, minimizing
$\sum_{i}f_{i}\left(\beta_{i}\right)$ is equivalent to minimizing
$\mbox{Trace}\left(\Gamma\right)$. Thus, the optimization problem
in (\ref{eq:OptimalConvexVaccination}) can be written as the following
semidefinite program (SDP):
\begin{eqnarray}
T^{*}\triangleq\min_{\Gamma} & \mbox{trace}\left(\Gamma\right)\nonumber \\
s.t. & \left(\delta-\varepsilon\right)\Gamma-A_{\mathcal{G}}\succeq0\nonumber \\
 & \bar{\beta}_{i}^{-1}\leq\gamma_{i}\leq\underline{\beta}_{i}^{-1}, & i=1,\ldots,n,\label{eq:OptimalConvexVaccination-1}
\end{eqnarray}
Given our network with 247 nodes, we now compute the optimal distribution
of vaccinations in several cases.

The network under consideration has a maximum eigenvalue $\lambda_{1}\left(A_{\mathcal{G}}\right)=13.52$.
In our simulations, we choose the value of $\bar{\beta}_{i}$ to induce
instability of the disease-free equilibrium in the absence of vaccination.
According to Proposition \ref{prop:HomogeneousSISstability}, if we
had a constant infection rate $\beta_{i}=\beta$ satisfying $\beta>\beta_{c}\triangleq\delta/\lambda_{1}\left(A_{\mathcal{G}}\right)=7.4e-3$,
the disease-free equilibrium is unstable. Hence, we choose a natural
infection rate $\bar{\beta}_{i}=\bar{\beta}>\beta_{c}$ to induce
an unstable infection in the absence of vaccinations. In our simulations,
individuals have the same natural infection rates $\bar{\beta}_{i}=\bar{\beta}$,
and study three cases: $\bar{\beta}\in\left\{ 1.2\beta_{c},1.8\beta_{c},2.4\beta_{c}\right\} $.
We choose the value of $\underline{\beta}_{i}<\beta_{c}$
to induce a stable disease-free equilibrium in the case of full-force
vaccination, i.e., we saturate all the individuals with vaccines to
shift their infection rates to $\underline{\beta}_{i}$. In our simulations
we use a minimum infection rate $\underline{\beta}_{i}=0.2\bar{\beta}_{i}=0.2\bar{\beta}$;
hence, we obtain that $\underline{\beta}_{i}\in\left\{ 0.24\beta_{c},0.36\beta_{c},0.48\beta_{c}\right\} $.
In other words, our vaccine reduces the infection rate to a 20\% of
the natural infection rate. Using these parameter values, we run three simulations,
each one with a different $\bar{\beta}$.

The results of our simulations are summarized in Fig. \ref{fig_2}. Each one
of the subplots in this figure correspond to a different value of
$\bar{\beta}\in\left\{ 1.2\beta_{c},1.8\beta_{c},2.4\beta_{c}\right\} $.
For each value of $\bar{\beta}$ we present a scatter plot with 247 data points (as many as individuals in the network),
where each point has an abscissa equal to $f_{i}\left(\beta_{i}\right)$
(the cost of vaccinating node $i$ with optimal fraction $\beta_i$) and an ordinate of $d_{i}$ (the
degree of $i\in\mathcal{V}(\mathcal{G})$). We observe that there
is a strong dependence between the cost of vaccinating a node and
its degree. In particular, we observe that there is almost an affine
relationship between the vaccination cost and the degree of a node,
with a saturation at the extreme cost values, 0 and 1. Also, we observe
that, as we increase the value of $\bar{\beta}$ the vaccination costs
tend to increase. This is because for larger $\bar{\beta}$, the
more virus is more infectious.

\section{Combinatorial Resource Allocation\label{sec:Combinatorial-Resource-Allocatio}}

In this Section, we consider a combinatorial versions of the fractional
vaccination problem studied in the previous section. In the fractional
vaccination problem, the optimal distribution of vaccines is allowed
to be in the feasible interval $\beta_{i}\in\left[\underbar{\ensuremath{\beta}}_{i},\bar{\beta}_{i}\right]$. In the combinatorial vaccination problem, we restrict the resources
to be in the discrete set, $\beta_{i}\in\left\{ \underbar{\ensuremath{\beta}}_{i},\bar{\beta}_{i}\right\} $.
For this case, we propose a greedy approach that provides an approximation
to the optimal combinatorial solution. We also provide quality guarantees
for this approximation algorithm in Subsection \ref{sub:Quality-Guarantee}.
The combinatorial vaccination problem can be stated as follows:
\begin{problem}
\emph{\label{Problem:CombinatorialVaccineDistribution}Given a curing
rate profile, $\left\{ \delta_{i}:i\in\mathcal{V}\left(\mathcal{G}\right)\right\} $,
and a vaccination cost function $f_{i}\left(\beta_{i}\right)$} for
$\beta_{i}\in\left\{ \underline{\beta}_{i},\bar{\beta}_{i}\right\} $,
\emph{find the optimal distribution of vaccines to control the propagation
of an epidemic outbreak with an asymptotic exponential decaying rate
$\varepsilon$ at a total minimum cost.}
\end{problem}
The optimal distribution of vaccines in Problem \ref{Problem:CombinatorialVaccineDistribution}
can be characterized by the set of individuals $I_{C}\subseteq\mathcal{V}\left(\mathcal{G}\right)$
that are chosen to be fully immunized, i.e., the infection rates are
switched from $\bar{\beta}_{i}$ to $\underline{\beta}_{i}<\bar{\beta}_{i}$
for $i\in I_{C}$. Let us assume that the vaccination cost function
takes the values $f_{i}\left(\bar{\beta}_{i}\right)=0$ and $f_{i}\left(\underline{\beta}_{i}\right)=c_{i}$.
These extreme values are achieved using the following affine cost
function
\[
f_{i}\left(\beta_{i}\right)\triangleq c_{i}\frac{\beta_{i}-\bar{\beta}_{i}}{\underline{\beta}_{i}-\bar{\beta}_{i}}.
\]
Hence, the total cost of vaccination satisfies
\[
\sum_{i=1}^{n}f_{i}\left(\beta_{i}\right)=a_{C}\sum_{i}c_{i}\beta_{i}-b_{C},
\]
where we have defined the constants $a_{C}\triangleq\left(\underbar{\ensuremath{\beta}}_{i}-\bar{\beta}_{i}\right)^{-1}$
and $b_{C}\triangleq a_{C}\sum_{i}c_{i}\bar{\beta}_{i}$. Thus, since
$a_{C}<0$, the optimal allocation of vaccines that minimizes $\sum_{i=1}^{n}f_{i}\left(\beta_{i}\right)$
is the same as the one that maximizes $\sum_{i}c_{i}\beta_{i}$. Therefore,
defining the vectors $c\triangleq\left(c_{1},\ldots,c_{n}\right)^{T}$
and $b\triangleq\left(\beta_{1},\ldots,\beta_{n}\right)^{T}$, Problem
\ref{Problem:CombinatorialVaccineDistribution} can be stated as the
following optimization problem:

\begin{eqnarray}
T_{C}^{*}=\max_{\left\{ \beta_{i}\right\} } & c^{T}b\nonumber \\
s.t. & \lambda_{1}\left(BA_{\mathcal{G}}-D\right)\leq-\varepsilon\label{eq:OptimalVaccine-1}\\
 & \beta_{i}\in\left\{ \underline{\beta}_{i},\bar{\beta}_{i}\right\} , & i=1,\ldots,n.\nonumber 
\end{eqnarray}
The solution to this problem is combinatorial in nature. In the following
subsections we provide a greedy approach that approximates the combinatorial
solution, as well as a quality guarantee of our approach.

\begin{figure*}
\centering
\newcolumntype{g}{>{\columncolor{Gray}}c}
\begin{tabular}{ | l | c || c| g | c | c |c|}
\hline
Parameters & Metric& Greedy& Reverse Greedy& Degree Threshold& Centrality Threshold& $D^*$\\
\hline\hline
$\bar\beta = 2.4 \beta_c$ & $c'b$&3.6298	 &3.6440&3.2892&2.4518&3.9425\\
\cline{2-7}
$\underline\beta = 0.2 \bar\beta$ & $\lambda_1(\delta B^{-1}-A)$&0.0054&0.0355&0.0422&0.1982&n/a \\
\hline \hline
$\bar\beta = 1.8 \beta_c$ & $c'b$& 3.0098&3.0098&2.9246&2.0092&3.1406\\
\cline{2-7}
$\underline\beta = 0.2 \bar\beta$ & $\lambda_1(\delta B^{-1}-A)$ &0.0850&0.1383&0.0774&0.2575&n/a\\
\hline \hline
$\bar\beta = 1.2 \beta_c$ & $c'b$ & 2.1484&2.1484&2.1201&1.7369&2.1787\\
\cline{2-7}
$\underline\beta = 0.2 \bar\beta$ & $\lambda_1(\delta B^{-1}-A)$ &0.4383&0.4383&0.6278&1.0101&n/a\\
\hline
\end{tabular}
\caption{Table with values of the objective function $c'b$ and the
residual value of $\lambda_{1}\left(\delta B^{-1}-A_{\mathcal{G}}\right)$
for each possible value of $\bar{\beta}_{i}$. \label{table}}
\end{figure*}
\subsection{Greedy approach\label{sub:Greedy-approach}}

In this subsection, we provide a greedy algorithm that iteratively
updates the set of nodes that will be (fully) vaccinated in order
to control the spreading of an epidemic outbreak. In each step of
our algorithm, we denote the set of nodes that are chosen to be part
of the vaccination group at $S_{t}$. We iteratively add to this group
the node that provides the most benefit per unit cost, where
the benefit of vaccinating a is the increment it induces in $\lambda_{1}\left(BA_{\mathcal{G}}-D\right)$.
More formally, given a vaccination group $S_{t}$, we define the diagonal
matrix of associated infection rates as $B_{S_{t}}\triangleq diag\left(\bar{\beta}_{i}\right)-(\bar{\beta}_{i}-\underline{\beta}_{i})diag(\boldsymbol{1}_{S_{t}})$,
where $\boldsymbol{1}_{S_{t}}$ is the $n$-dimensional indicator
vector for the set $S_{t}$. Thus, the benefit per unit cost of adding
node $i$ to $S_{t}$ is measured by the function
\[
\Delta\left(i,S_{t}\right)\triangleq\frac{\lambda_{1}\left(B_{S_{t}}A_{\mathcal{G}}-D\right)-\lambda_{1}\left(B_{S_{t}+\left\{ i\right\} }A_{\mathcal{G}}-D\right)}{c_{i}}.
\]
A conventional greedy approach could be defined by the iteration $S_{t+1}=S_{t}+\left\{ i_{t}\right\} $
with $S_{1}=\left\{ \right\} $ and $i_{t}\triangleq\arg\max_{i}\Delta\left(i,S_{t}\right)$,
where this iteration is repeated until $\lambda_{1}\left(B_{S_{t}}A_{\mathcal{G}}-D\right)\leq-\varepsilon$
is satisfied. Notice that the resulting vaccination group is feasible
and satisfies the spectral condition needed to control the spreading
of an epidemic outbreak. 

In practice, we observe that a modification of this greedy approach
provides better results. In this modified version, we start with a
vaccination set $S_{1}=\mathcal{V}\left(\mathcal{G}\right)$ (i.e.,
all the individuals are vaccinated) and iteratively remove individuals
according to the iteration $S_{t+1}=S_{t}-\left\{ j_{t}\right\} $
with $j_{t}=\arg\min_{j}\Delta\left(j,S_{t}\backslash\left\{ j\right\} \right)$,
where this iteration is repeated until $\lambda_{1}\left(B_{S_{t}}A_{\mathcal{G}}-D\right)\geq-\varepsilon$
is satisfied. The final vaccination group is chosen to be $S_{t-1}$.
Notice that, the resulting vaccination group is feasible and $\lambda_{1}\left(B_{S_{t-1}}A_{\mathcal{G}}-D\right)\leq-\varepsilon$.
We denote this approach the \emph{reverse greedy algorithm}.

Since our approach is heuristic for a combinatorial problem, we provide
a quality guarantee via Lagrange duality theory in the following subsection.

\subsection{Quality Guarantee\label{sub:Quality-Guarantee}}
Using Lagrange duality theory, we provide quality guarantees for
the performance of our greedy approach by computing the dual optimal $D_C^*$.
\begin{theorem} \label{SDP} Given the optimization problem 
\begin{eqnarray}
T_{C}^{*}= & \max_{b} & c^{T}b\label{primal}\\
 & \hbox{s.t.} & \left(D-\varepsilon I\right)B^{-1}-A_{\mathcal{G}}\succeq0\nonumber \\
 &  & \beta_{i}\in\{\underline{\beta}_{i},\bar{\beta}_{i}\},\mbox{ }\forall i,\nonumber 
\end{eqnarray}
the primal optimal $T_{C}^{*}$ can be upper bounded by $D_{C}^{*}$
computed according to the Lagrange dual 
\begin{eqnarray}
D_{C}^{*}= & \min_{Z,u} & \mathbf{1}^{T}u-trace(A_{\mathcal{G}}Z)\label{dual}\\
 & \hbox{s.t.} & u_{i}\ge c_{i}\bar{\beta}_{i}+\frac{\delta_{i}}{\bar{\beta}_{i}}Z_{ii}\,\forall i\nonumber \\
 &  & u_{i}\ge c_{i}\underline{\beta}_{i}+\frac{\delta_{i}}{\underline{\beta}_{i}}Z_{ii}\,\forall i\nonumber \\
 &  & Z\succeq0,\nonumber 
\end{eqnarray}
which is a convex Semidefinite Program. \end{theorem} \begin{proof}
Notice that matrix in the semidefinite constrain can be written as
$\left(D-\varepsilon I\right)B^{-1}-A_{\mathcal{G}}=\sum_{i}e_{i}e_{i}'\frac{\delta_{i}-\varepsilon}{\beta_{i}}-A_{\mathcal{G}}$,
where $e_{i}$ is the unit vector in the standard basis. From \eqref{primal},
we construct the Lagrangian 
\begin{equation}
\mathcal{L}(b,Z)=c^{T}b+trace\left(Z\left(\sum_{i}e_{i}e_{i}'\frac{\delta_{i}}{\beta_{i}}-A_{\mathcal{G}}\right)\right),
\end{equation}
where $\beta_{i}\in\{\underline{\beta}_{i},\bar{\beta}_{i}\}$ is
kept as a domain constraint and $Z\succeq0$. See Section 5.9 of \cite{BV04}
for further details on the Lagrange dual of semidefinite constraints.
Using the properties of trace to simplify and decouple we get 
\begin{equation}
\mathcal{L}(b,Z)=\sum_{i}\left(c_{i}\beta_{i}+\frac{\delta_{i}}{\beta_{i}}Z_{ii}\right)-trace(ZA_{\mathcal{G}}).\label{lagrange}
\end{equation}
The dual objective is derived by maximizing the Lagrangian with respect
to the primal variables 
\begin{equation}
q(Z)=\sum_{i}\left(\max_{\beta_{i}}c_{i}\beta_{i}+\frac{\delta_{i}}{\beta_{i}}Z_{ii}\right)-trace(ZA_{\mathcal{G}}).\label{q}
\end{equation}
Due to the decoupling in \eqref{lagrange} the primal optimization
in \eqref{q} can be done for each node, independently. Since each
node has only 2 options we can consider each case explicitly by defining
\begin{equation}
u_{i}=\max\left\{ c_{i}\bar{\beta}_{i}+\frac{\delta_{i}}{\bar{\beta}_{i}}Z_{ii},c_{i}\underline{\beta}_{i}+\frac{\delta_{i}}{\underline{\beta}_{i}}Z_{ii}\right\} .\label{t}
\end{equation}
It is possible to compute $u_{i}$ as a threshold function of $Z_{ii}$,
but for the purpose of constructing the dual it is better to use an
epigraph formulation to rewrite \eqref{q} as 
\begin{equation}
q(Z,u)=\sum_{i}u_{i}-trace(ZA_{\mathcal{G}})\label{q2}
\end{equation}
with the addition constraints that 
\begin{eqnarray}
u_{i} & \ge & c_{i}\bar{\beta}_{i}+\frac{\delta_{i}}{\bar{\beta}_{i}}Z_{ii}\label{tbar}\\
u_{i} & \ge & c_{i}\underline{\beta}_{i}+\frac{\delta_{i}}{\underline{\beta}_{i}}Z_{ii}.\label{tunder}
\end{eqnarray}
Since the dual is a minimization and $q(Z,u)$ is strictly increasing
in $u$, either \eqref{tbar} or \eqref{tunder} must be achieved
with equality, ensuring that the definition \eqref{t} is satisfied
at the optimal point. To conclude, our dual \eqref{dual} is given
by minimizing \eqref{q2} subject to the domain constraint $Z\succeq0$
and the epigraph constraints \eqref{tbar} and \eqref{tunder}. This
is a standard form SDP as defined in section 4.6 of \cite{BV04}.
The solution $D_{C}^{*}$ is guaranteed to satisfy $D_{C}^{*}\ge T_{C}^{*}$
by weak duality, \cite{BV04} Section 5.2. \end{proof}

Theorem \ref{SDP} tells us that for any optimization problem of the
form \eqref{primal} we can get an accuracy certificate 
\begin{equation}
T_{C}^{*}-c^{T}b\le D_{C}^{*}-c^{T}b
\end{equation}
by solving the dual \eqref{dual}. Since we do not have a strong duality,
we do not expect $c^{T}b=D_{C}^{*}$ to be attainable (i.e, $P_C^* < D_C^*$).

\begin{remark} The solution to the dual gives us some insight into
the primal optimizers via the threshold solution to \eqref{t}, 
\begin{equation}
u_{i}(Z_{ii})=\left\{ \begin{array}{ll}
c_{i}\bar{\beta}_{i}+\frac{\delta_{i}}{\bar{\beta}_{i}}Z_{ii} & \hbox{if }Z_{ii}\le\frac{c_{i}}{\delta_{i}}\bar{\beta}_{i}\underline{\beta}_{i}\\
c_{i}\underline{\beta}_{i}+\frac{\delta_{i}}{\underline{\beta}_{i}}Z_{ii} & \hbox{if }Z_{ii}\ge\frac{c_{i}}{\delta_{i}}\bar{\beta}_{i}\underline{\beta}_{i}
\end{array}\right..\label{t2}
\end{equation}
It appears we can deduce the primal optimizers $b^{*}$ from $Z^{*}$,
but in practice for most nodes $i$, $Z_{ii}^{*}=\bar{\beta}_{i}\underline{\beta}_{i}c_{i}/\delta_{i}$
making it impossible to determine $\beta_{i}^{*}$. In some cases
there are nodes that have $Z_{ii}^{*}$ not equal to the threshold.
These nodes have their optimal action specified by $Z_{ii}^{*}$ and
\eqref{t2}. This at least allows for a reduction of the dimension
of the primal problem which due to its combinatorial form could be
a very large improvement. \end{remark} 

\subsection{Numerical Results}
\begin{figure}[t]
\centering\includegraphics[width=0.45\textwidth]{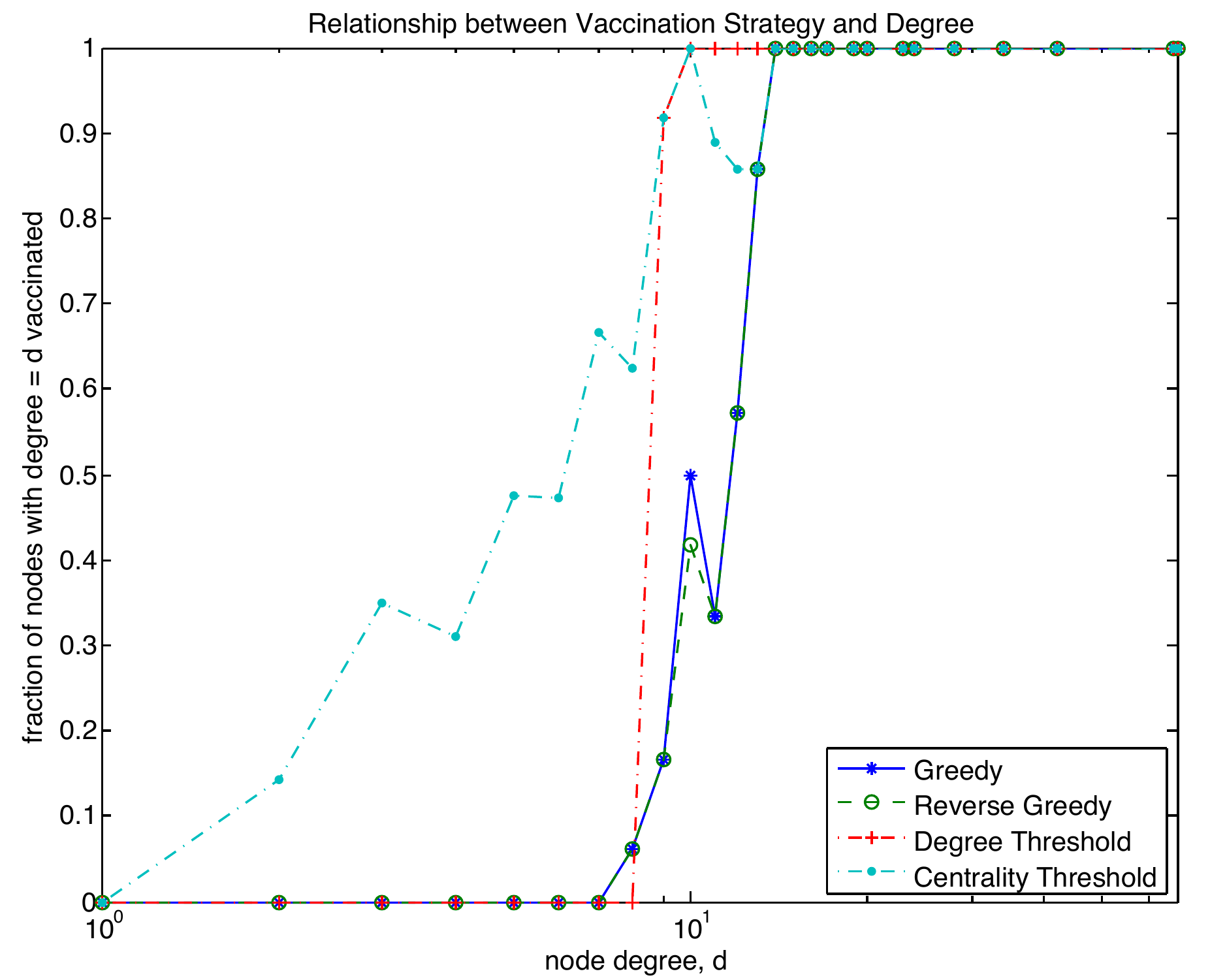}
\caption{the relationship between the outcome
of each algorithm and the degrees of the nodes. We represent degrees in log scale versus the fraction of nodes of a particular degree that are vaccinated in
the solution generated by each algorithm.}
\label{fig_3}
\end{figure}

Several papers in the literature advocate for vaccination strategies
based on popular centrality measures, such as the degree or eigenvector
centrality \cite{CHT09}. In this subsection, we compare our greedy
heuristic to vaccination strategies based on centrality measures.
In our simulation, we use the adjacency matrix with 247 nodes previously
used in Subsection \ref{sub:Numerical-Results} and the same values
for the parameters $\delta_{i}=\delta=0.1$, $\beta_{c}=\delta/\lambda_{\hbox{max}}(A_{\mathcal{G}})=7.4e-3$,
$\bar{\beta}_{i}\in\left\{ 1.2\beta_{c},1.8\beta_{c},2.4\beta_{c}\right\} $
and$\underline{\beta}_{i}=0.2\bar{\beta}_{i}$ for all $i$. In Table
\ref{table}, we include the values of the objective function $c'b$ and the
residual value of $\lambda_{1}\left(\delta B^{-1}-A_{\mathcal{G}}\right)$
for each possible value of $\bar{\beta}_{i}$. In each case, we run
the greedy algorithm and the reverse greedy algorithm (both proposed
in Section \ref{sub:Greedy-approach}), as well as two previously
proposed algorithms based on the degree and the eigenvalue centrality
metrics. In the last column of Table \ref{table}, we also include the upper
bound provided by Theorem \ref{SDP}. Observe that our greedy algorithms
are always within 10\% of the upper bound $D_{C}^{*}$. Furthermore,
the reverse greedy algorithm is outperforms the others, specially
those based on centrality measures. 

In Fig. \ref{fig_3}, we illustrate the relationship between the outcome
of each algorithm and the degrees of the nodes. In the abscissae,
we represent degrees in log scale, and in the ordinate we provide
the fraction of nodes of a particular degree that are vaccinated in
the solution generated by each algorithm. We observe how all four
algorithms completely vaccinate the set nodes with degrees beyond
a threshold. On the other hand, in the range of intermediate degrees,
we observe that degree alone is not sufficient information to decide
the vaccination level of a node. In other words, simply vaccinating nodes based on degree does not always provide the best results.

\begin{figure}[t]
\centering\includegraphics[width=0.45\textwidth]{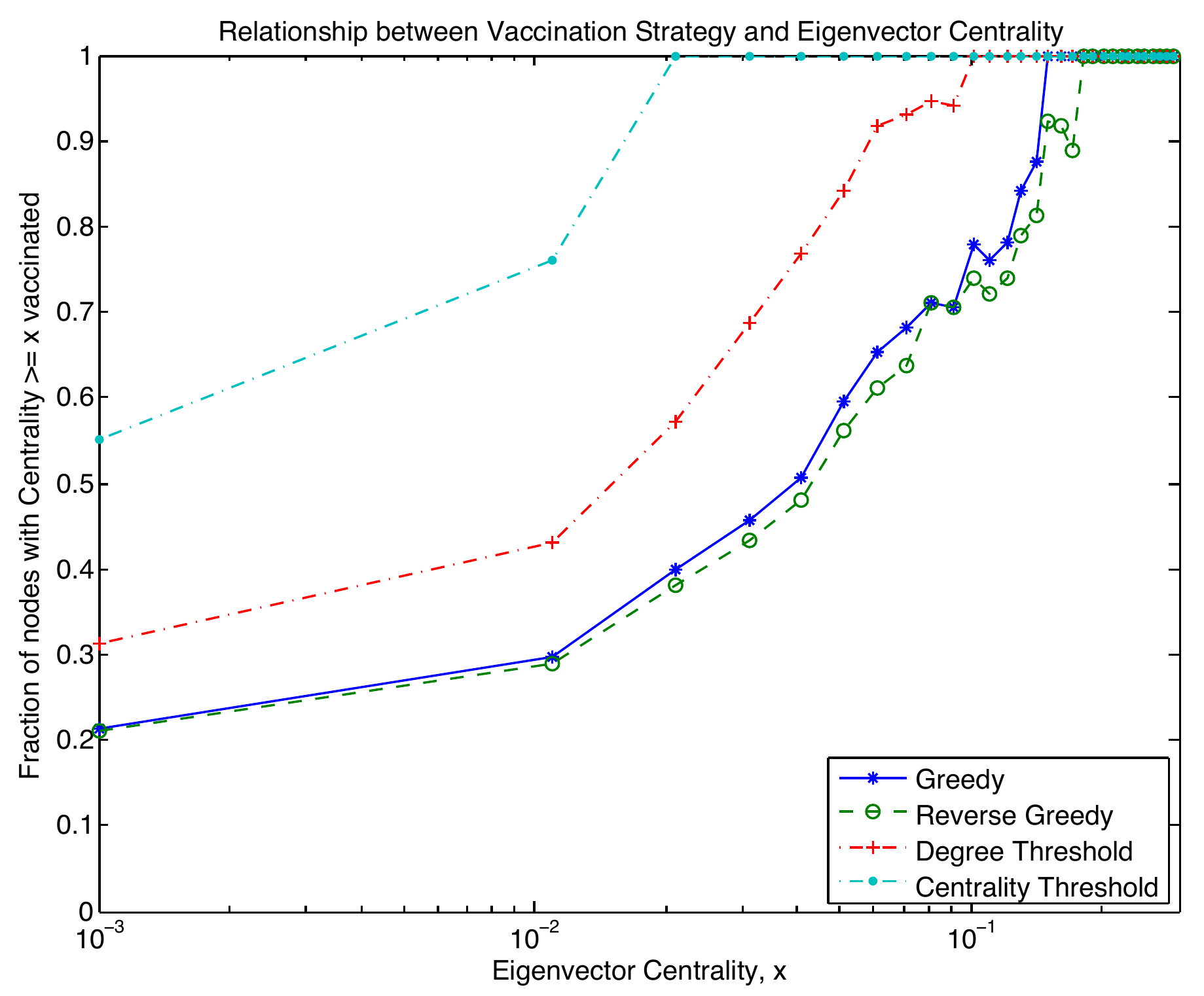}
\caption{Relationship between the outcome
of each algorithm and the eigenvector centrality. In the abscissae,
we represent the cumulative fraction of nodes with centrality greater
or equal to a given value being vaccinated in the outcome of each
one of the four algorithms under consideration.}
\label{fig_4}
\end{figure}

In Fig. \ref{fig_4}, we illustrate the relationship between the outcome
of each algorithm and the eigenvector centrality. In the abscissae,
we represent the cumulative fraction of nodes with centrality greater
or equal to a given value being vaccinated. We observe how all
four algorithms completely vaccinate the set nodes with the highest centralities. However, since the curves in this figure are
not monotonically increasing, there must be cases in which lower centrality
nodes are vaccinated, but other nodes with higher centrality are left
unvaccinated. In other words, vaccinating higher centrality nodes
does not always provide the best results.

The reason neither degree nor centrality adequately capture the importance
of nodes is that the eigenvectors of the matrix $\delta B^{-1}-A_{\mathcal{G}}$
change when the set of vaccinated nodes change. The shift in the eigenvectors
is a result of the fact that optimal vaccination strategy actually depends on the parameters
$\bar{\beta}$ and $\underline{\beta}$, not just the network. With this in mind, we cannot expect an optimal
solution to arise from an algorithm that depends only on the graph
structure. Our algorithms work because they are greedy with respect
to this effect. 

\section{Conclusions}

We have studied the problem of controlling the dynamic of the SIS
epidemic model in an arbitrary contact network by distributing vaccination
resources throughout the network. Since the spread of an epidemic
outbreak is closely related to the eigenvalues of a matrix that depends
on the network structure and the parameters of the model, we can formulate
our control problem as a spectral optimization problem in terms of
semidefinite constraints. In the fractional vaccination case, where
intermediate level of vaccination are allowed, we have proposed a
convex optimization framework to efficiently find the optimal allocation
of vaccines when the function representing the vaccination cost satisfies
certain convexity assumptions. In the combinatorial vaccination problem,
where individuals are not allowed to be partially vaccinated,
we propose a greedy approach with quality guarantees based on Lagrangian
duality. We illustrate our results with numerical simulations in a
real online social network.

\end{document}